\documentclass[sigconf,authorversion,nonacm]{acmart}

\AtBeginDocument{%
  \providecommand\BibTeX{{%
    \normalfont B\kern-0.5em{\scshape i\kern-0.25em b}\kern-0.8em\TeX}}}

\usepackage{graphicx}
\usepackage{subcaption}
\usepackage{soul}
\usepackage{tikz}
\usetikzlibrary{arrows}
\usetikzlibrary{backgrounds}
\tikzstyle{int}=[draw, fill=blue!20, minimum size=2em]
\tikzstyle{init} = [pin edge={to-,thin,black}]
\usepackage{pgfplots}
\pgfplotsset{compat=newest}
\usepackage[ruled,vlined]{algorithm2e}
\usepackage{enumitem}
\setlist[itemize] {noitemsep,nolistsep,leftmargin=*}

\newcommand{\squishlisttwo}{
   \begin{list}{$\bullet$}
       { \setlength{\itemsep}{0pt}      \setlength{\parsep}{3pt}
       \setlength{\topsep}{3pt}       \setlength{\partopsep}{0pt}
      \setlength{\leftmargin}{1.5em} \setlength{\labelwidth}{1em}
       \setlength{\labelsep}{0.5em} } }

\newcommand{\squishlistthree}{
   \begin{list}{$\bullet$}
   { \setlength{\itemsep}{0pt}    \setlength{\parsep}{0pt}
   \setlength{\topsep}{0pt}     \setlength{\partopsep}{0pt}
   \setlength{\leftmargin}{2em} \setlength{\labelwidth}{1.5em}
   \setlength{\labelsep}{0.5em} } }
\newcommand{\squishend}{
\end{list}  }

\newcommand{\etal}{et al. }

\usepackage{comment}

\begin{document}

\title{Multi-Agent Join}


\author{Vahid Ghadakchi}
\affiliation{
	\institution{School of EECS\\Oregon State University}
        \country{USA}
}
\email{ghadakcv@oregonstate.edu}

\author{Mian Xie}
\affiliation{
	\institution{School of EECS\\Oregon State University}
        \country{USA}
}
\email{xiemia@oregonstate.edu}

\author{Arash Termehchy}
\affiliation{
	\institution{School of EECS\\Oregon State University}
        \country{USA}
}
\email{termehca@oregonstate.edu}

\author{Bakhtiyar Doskenov}
\affiliation{
	\institution{School of EECS\\Oregon State University}
        \country{USA}
	}
\email{doskenob@oregonstate.edu}

\author{Bharghav Srikhakollu}
\affiliation{%
	\institution{School of EECS\\Oregon State University}
        \country{USA}
}
\email{srikhakb@oregonstate.edu}

\author{Summit Haque}
\affiliation{
	\institution{School of EECS\\Oregon State University}
        \country{USA}
}
\email{haquesu@oregonstate.edu}

\author{Huazheng Wang}
\affiliation{%
	\institution{School of EECS\\Oregon State University}
        \country{USA}
}
\email{huazheng.wang@oregonstate.edu}



\begin{abstract}
It is crucial to provide real-time performance in many applications, such as interactive and exploratory data analysis. 
In these settings, users often need to view subsets of query results quickly. 
It is challenging to deliver such results over large datasets for relational operators over multiple relations, such as join. 
Join algorithms usually spend a long time on scanning and attempting to join parts of relations that may not generate any result.
Current solutions usually require lengthy and repeating preprocessings, which are costly and may not be possible to do in many settings.
Also, they often support restricted types of joins.
In this paper, we outline a novel approach for achieving efficient join processing in which a scan operator of the join learns during query execution, the portions of its relations that might satisfy the join predicate. 
We further improve this method using an algorithm in which both scan operators collaboratively learn an efficient join execution strategy.
We also show that this approach generalizes traditional and non-learning methods for join.
Our extensive empirical studies using standard benchmarks indicate that this approach outperforms similar methods considerably. 
\end{abstract}



\maketitle
\section{Introduction}
\label{sec:intro}


It is crucial to provide real-time performance for queries over large data in many settings, such as interactive and exploratory data analysis and online query processing \cite{Carey:1997:SLA:253260.253302,DBLP:conf/vldb/CareyK98,Hellerstein:1997:OA:253260.253291,Haas:1999:RJO:304182.304208,li2016wander}. 
However, it is very challenging to provide highly efficient query processing over large data. 
In these settings, users can often extract accurate insights using only subsets of query results \cite{Hellerstein:1997:OA:253260.253291,Haas:1999:RJO:304182.304208,li2016wander}.
Thus, to satisfy the desired response time, researchers have proposed {\it approximate query processing} methods that estimate desired insight quickly using a relatively small sample of query results.
Using this approach, the system may return progressively more accurate query results so users can stop the execution as soon as they have enough information \cite{Hellerstein:1997:OA:253260.253291,Haas:1999:RJO:304182.304208,li2016wander,DBLP:conf/sigmod/JermaineDAJP05}. 

This is, however, challenging to return sufficiently accurate results quickly for queries over more than a single relation.
Examples of these queries are \textit{join} or {\it intersection} of relations. 
To process these queries, an algorithm may explore and check a huge space of possible combinations of tuples from different relations, e.g., quadratic space for a join of two relations.
Since most combinations of tuples do not satisfy the condition of a given query, e.g., one with a join predicate, a significant majority of tested combinations do {\it not} return any result. 

Researchers have proposed algorithms that return subsets of query results quickly, however, these algorithms often require substantial amounts of computational resources, significant time and effort for preprocessing and maintenance, or handle only restricted types of queries \cite{Hellerstein:1997:OA:253260.253291,Haas:1999:RJO:304182.304208,DBLP:conf/sigmod/JermaineDAJP05,li2016wander}. 
For instance, {\it Ripple Join} progressively reads tuples from base relations and keeps them in main memory to improve efficiency of the query execution \cite{Haas:1999:RJO:304182.304208}.
It, however, often runs out of main memory for large relations before generating sufficiently many results \cite{DBLP:conf/sigmod/JermaineDAJP05}.

Some methods build auxiliary data structures, e.g., indexes, over relations to locate tuples that generate results quickly \cite{li2016wander}.
Most users, however, cannot afford to wait for the time-consuming preprocessing steps of building these data structures, which are often repeated whenever the data evolves.
Indexes also have substantial storage and update overheads.
Moreover, required indexes for a query workload are usually determined by experts or pre-trained models \cite{chaudhuri2007self,vanaken21}. 
Normal users do not have the expertise to guide index building effectively. 
There might not also be sufficient training data to learn an accurate model of the query workload to guide these preprocessing steps.
Furthermore, this approach cannot handle many types of queries, e.g., joins with user-defined functions.

Some methods use sort- and hash-based approaches to reorganize the input relations so they can quickly find subsets of tuples that match the query conditions \cite{DBLP:conf/sigmod/JermaineDAJP05,journals/debu/UrhanF00,10.1145/564691.564721}. 
These methods, however, support only limited types of queries, e.g., cannot handle joins with user-defined functions.
Due to the popularity of statistical analysis and learning over large data, join predicates increasingly contain complex conditions, such as user-defined functions.
They might also need a significant amount of main memory to generate results efficiently, e.g., sort-based methods keep considerable portions of relations in main memory \cite{DBLP:conf/sigmod/JermaineDAJP05}.

Toward realizing the vision of delivering real-time insight over large data, in this paper, we propose a novel approach to join processing that {\it learns online and during query execution} to return join results quickly.
Based on the results of early query processing, our method learns the portion of the (base) relations that are most likely to produce the most output tuples and joins them.
Our approach does not require any pre-processing or indexes.
It views the join predicate as a black box and handles joins with complex conditions. 

Nevertheless, it is challenging to accurately learn efficient query processing strategies over large data and produce query results quickly at the same time.
Our method should establish a trade-off between producing fresh results based on its current information and strategy, i.e., {\it exploitation}, and searching for more efficient strategies, i.e., {\it exploration}.
Moreover, as each relation may contain numerous tuples, it must explore many possible strategies to find the efficient ones. 
This becomes more challenging by the usual restrictions on accessing data on the secondary storage.
For example, in the absence of any index, it must examine tuples on disk sequentially. 
We leverage tools from the area of online learning and devise algorithms that address the aforementioned challenges efficiently with theoretical guarantees. 
Our contributions:

\squishlisttwo
    \item We provide a novel framework for processing a join query as collaborative online learning of its scan operators. 
    We show that this approach generalizes the current methods of join processing. Therefore, a system might use a combination of traditional algorithms and learning operators to execute queries (Section~\ref{qp:framework}).
    
    \item  We leverage popular online learning methods to design an algorithm for scan operators that learns portions of base relation that generate the most results. But, we show that due to the enormous search space, restrictions in access to data, and limited rewards, scan operators cannot use popular online learning methods to learn efficient strategies quickly (Section~\ref{qp:challenges}). 
        
    \item We propose a novel and efficient algorithm to learn efficient strategies in which one scan operator learns its strategy and the other one uses sequential scan (Section~\ref{sec:singleLearning}).
        
    \item Users often would like to compute accurate estimates of aggregation functions over sample results. Our proposed method may be biased toward the results that are created from the portions of base relations that produce more joins than other parts of the relations.  
    We propose an unbiased, consistent, and asymptotically normal estimator to compute aggregation functions over the results of our method, which delivers high-probability confidence intervals for its estimates (Section~\ref{sec:estimate}).
    
    \item We show that it will improve the time of producing subsets of join results, if both scan operators of join learn and adapt their strategies online.
        The resulting join algorithms, however, may perform redundant data access and computation. 
        We propose techniques so that the scan operators coordinate their learning to reduce these redundant accesses and collaboratively learn efficient query processing strategies  (Section~\ref{sec:collaboration}).
    
    \item We report the results of our extensive empirical study using a standard benchmark. 
        Our empirical studies indicate that the proposed methods significantly outperform comparable current techniques for delivering subsets of results in most cases.
        We show that where learning is not by definition helpful, e.g., all tuples in both relation generate the same number of results, our method is almost as efficient as comparable methods (Section~\ref{qp:empirical}). 
\squishend

One may extend our methods to other types of relational operators, such as intersection. 
However, due to the popularity of joins, in this paper, we focus on and report empirical studies for joins.
\section{Framework}
\label{qp:framework}
We present the framework that models processing of a join as online collaborative learning of its operators. 
We first describe our framework and methods for binary joins and then explain how to use our approach for queries with more relations in Section~\ref{sec:collaboration}. 

\subsection{Agents, Actions, \& Reward}
\noindent{\bf Agents and Actions.}
Relational data systems usually model a query as a tree of a logical operators called {\it logical query plan} \cite{DBBook}. 
In a join query, the leaf nodes are scan operators that read information from the base relations. 
The tuples are pipelined from scans to the root operator, i.e., join, as the join operator calls its children's API to get fresh tuples. 
Each \textit{iteration} (or \textit{round}) of processing the query starts with the call from the join operator.
In our framework, each scan operator in the query plan is a learning {\it agent}. 
The {\it actions} of an agent is the set of its available activities in each round. 
The actions of each scan operator are the set of tuples in its base relation.

\noindent{\bf Reward.}
The {\it reward} of an action in each round of the query execution is the total number of joint tuples produced during that round. 
The join operator shares this reward with its child scan operators immediately after attempting to join their recently sent tuples. 
The {\it history} of a scan operator $O$ at round $t$ denoted as $H^{O}(t)$, is the sequence of pairs $(a_i,r_i)$, $0 \leq i \leq t-1$, where $a_i$ and $r_i$ are the action and the reward of the operator $O$ at round $i$ of the join.

\subsection{Strategy \& Learning}
\noindent{\bf Strategies.}
The {\it strategy} or policy of an operator $O$ at round $t$ is a mapping from $H^{O}(t)$ to the set of its available actions.
Each strategy of a (logical) operator is a possible execution algorithm.
An operator might follow a {\it fixed strategy} that does not change in the course of query execution. 
Current query operators often follow fixed strategies. 
For example, the scan operator for the outer relation in the nested loop join follows a fixed strategy of sending the next tuple (or page) of the relation stored on the disk whenever requested by its parent join operator. 
The scan operator for the inner relation follows a similar fixed strategy.

An operator may use an {\it adaptive strategy} and choose actions in each round based on the actions' rewards and performance in their previous rounds using the history up to the current round.
For example, scan operators for the outer and inner relations in the nested loop join may use their history to send the tuples that are more likely to produce new joint tuples instead of the ones that may not lead to any results. 
Using the history of the join, a scan operator may learn that tuple $t_1$ in its base relation joins with significantly more tuples of the other relation than tuple $t_2$, i.e., $t_1$ produces more results (reward) than $t_2$.
Thus, if it sends $t_1$ to the join operator more often than $t_2$, it may cause the query to produce more joint tuples given the same time limit. 
As the query execution progresses, the operators may learn and improve the effectiveness of their strategies using the performance of their actions in preceding rounds.

Our framework generalizes current approaches to query processing in which operators follow fixed strategies. 
Hence, it supports using both traditional and adaptive strategies for query operators.

\noindent{\bf Learning Strategies.}
Since the rewards of actions are not known at the start of the query processing, an operator has to learn these rewards while executing the query.
Thus, the operator may initially explore a subset of available actions to find the reasonably rewarding actions quickly. 
Subsequently, the operator may {\it exploit} this knowledge and use the most rewarding actions according to its investigations so far to increase the short-term reward or {\it explore} actions that have not been tried to find actions that produce higher rewards with the goal of improving the total reward in the long-run.
For example, consider the join of relations $R$ and $S$ where the scan over $S$ uses a fixed strategy of sending a randomly chosen tuple in each round. 
The scan over $R$ modifies its strategy in each round based on the information in the preceding rounds to read and send tuples  to the join operator that generate the most join results.  
It may initially send a few randomly chosen tuples to its parent join operator to both produce some joint tuples and estimate the reward of each tuple in this subset. 
In the subsequent rounds, the scan operator on $R$ may send the tuples with highest rewards so far, i.e., exploit, or pick other tuples that have not been tried before with the hope of finding ones with higher estimated reward, i.e., explore.
Arguably, the most important challenge in online learning is to find a balance between exploration and exploitation.

\noindent{\bf Objective Function.}
\label{sec:framework:objective}
A query processing algorithm should return the results of a query as fast as possible, therefore, each operator should maximize its total reward using fewest possible rounds. 
As explained in Section~\ref{sec:intro}, in approximate query processing, it is often desirable to deliver the results progressively where users investigate earlier results quickly while the system executes the query and delivers the rest of the results \cite{Carey:1997:SLA:253260.253302,Haas:1999:RJO:304182.304208,Hellerstein:1997:OA:253260.253291}.
To quantify the overall objective of query processing in this setting, one may select a metric, such as {\it discounted (weighted/geometric) average of delays}, that is biased toward faster generation of early results.
It measures the users' waiting time for receiving samples of results. 
More precisely, the discounted weighted average of delays is defined as $\sum_{i=1}^l{\gamma^it_i}$ where $0< \gamma < 1$, $t_i$ is the time to generate the $i$th result, and $l$ is set to the number of desired results.
It reflects the benefits of progressive or approximate query processing over large data more precisely than some other metrics, such as {\it completion time}.
The faster a query operator learns an efficient strategy during query execution, the smaller the value of this objective function is.

\noindent{\bf Why Dual-Agent Framework?}
One may implement the learning of the processing strategy for the entire query in a single place, e.g., the root join operator, instead of both scans and join operators.
However, each operator has access to a distinct subset of data, e.g., each scan operator reads a different relation, and consequently a separate subset of available actions. 
It is natural to place learning of an effective strategy to apply those actions in its operator.
Learning an effective strategy becomes substantially more difficult as the number of possible actions grows \cite{slivkins2019introduction}. 
Our framework simplifies learning effective strategies by distributing the set of possible actions and placing relevant actions of each operator in its relevant learning agent.
Moreover, in our dual-agent approach, each operator communicates with the other operator(s) by sending data and receiving rewards and treats their internal logic, e.g., join predicate, and implementation as black boxes. 
Thus, as opposed to the centralized approach, the implementations of learning operators may be reused across different queries, e.g., a learning scan can be used for binary operators other than join, such as intersection. 
This approach may also simplify using parallelism in query processing as done in traditional relational query processing.

\section{UCB-Scan}
\label{qp:challenges}
\noindent{\bf Multi-armed Bandit.}
{\it Multi-armed Bandit} (MAB) is a classic and arguably the most popular  model for online learning problems that formalizes the {\it exploration/ exploitation } trade-off explained in Section~\ref{qp:framework} \cite{slivkins2019introduction,lattimore2020bandit}.
Consider a set of actions, i.e., arms, with unknown expected rewards.
Every time an action is performed, we observe a random sample of its reward.
The objective in MAB is to find the optimal sequence of actions to maximize the total expected reward.
{\it Upper Confidence Bound (UCB)} is a popular group of algorithms for MAB \cite{slivkins2019introduction}. 
In a nutshell, they initially try every action once, observe its reward, and provide preliminary estimations of the expected reward of each action. 
In the subsequent rounds, they pick more often the actions that have {\it sufficiently large estimated expected rewards} or {\it have not been tried enough} and update their estimated rewards.
After sufficiently many trials, they commit to the single action with the largest estimated expected reward and pick it forever.
This approach delivers a statistically near-optimal balance between exploration and exploitation. 

\noindent{\bf Using Multi-armed Bandit Algorithms.}
One may use the aforementioned UCB algorithms to learn strategies for scan operators in a join. 
We denote the scan over $R$ and $S$ in the join of $R$ and $S$ as $R$-scan and $S$-scan, respectively. 
To simplify the exposition, we first assume that $S$-scan has a fixed strategy of reading a randomly chosen tuple from $S$ and sending it to the join operator, i.e., {\it random strategy}. 
It is shown that in the absence of any order, e.g., sorted relation, sequential scan, i.e., heap-scan, simulates random sampling effectively \cite{Haas:1999:RJO:304182.304208}. 
Thus, we assume that $S$-scan implements the random strategy using sequential scan over $S$. 
On the other hand, $R$-scan aims at \ul{learning the rewards of tuples in $R$ accurately and quickly} and joining tuples of $R$ with $S$ in a \ul{decreasing order of their rewards}. 
This way, the join operator progressively produces results efficiently and significantly reduces the users' total waiting time to view query results.
As $R$-scan does not know the reward of tuples of $R$ prior to executing the join, its learning should take place during join processing. 

To leverage UCB algorithms, $R$-scan first sequentially scans the entire $R$ once. 
For each tuple $t_R$ read from $R$, we sequentially read one tuple from $S$ and join these tuples to observe the reward of $t_R$. 
$R$-scan keeps track of rewards of each tuple of $R$ in a {\it reward table} maintained in main memory. 
Each entry in the reward table maintains the address of every read tuple in $R$ (on disk or memory) and its current accumulated reward.
After this step, $R$-scan randomly accesses the tuple deemed the most promising by the UCB algorithm and joins it with the next tuple of $S$. 
This may require multiple sequential scans of $S$.
We call this algorithm {\it UCB-Scan}. 
To avoid generating duplicates, the join operator keeps track of the tuples of $R$ and $S$ that have been joined.

\subsection{Shortcomings of UCB-Scan}
\label{qp:ShortcomingInitial}
\noindent{\bf Finite Rewards.}
UCB algorithms assume that (almost) every action produces some non-zero reward whenever it is tried, i.e., \ul{each action has infinite accumulated reward}.
However, the accumulated reward of each tuple in our setting is finite: after a tuple of $R$ is joined with every tuple of $S$ it will deliver zero reward forever.
Using only the estimated most rewarding tuple of $R$, the join may not produce sufficiently many results. 
Thus, as opposed to UCB methods, $R$-scan must detect an exhausted tuple and proceed to find the next most rewarding tuple in $R$.
More importantly, MAB algorithms often allow lengthy exploration. 
But, due to finite reward of each tuple, it may be worth stopping exploration altogether as soon as the algorithm finds a tuple with a sufficiently large estimated reward during exploration.
The algorithm may then exhaust that tuple by joining it with the rest of $S$ before trying other tuples in $R$.

\noindent{\bf Large Set of Actions.}
Each scan operator may face numerous possible actions at each step of query processing over large relations.
It often has to find the most rewarding tuple from millions of tuples in its relation. 
Following UCB algorithms, the UCB-Scan inspects every action many times before finding the most rewarding one.
Given the considerable overhead associated with accessing each tuple, e.g., I/O random access, it may take UCB-Scan too long to learn the most rewarding tuple.

\noindent{\bf Lack of Collaboration.}
In UCB-Scan, only one scan operator ($R$-scan) learns the optimal order for joining its tuples.
But, there may be several highly rewarding tuples in $S$ learning of which enables the system to generate many results quickly.
This may happen particularly when each tuple of $R$ joins with many tuples in $S$ and vice versa., i.e., $m$ to $n$ joins. 
Hence, $S$-scan may use UCB-Scan to learn rewarding tuples in $S$.
UCB algorithms require random samples of rewards for each action to deliver accurate reward estimations \cite{slivkins2019introduction}. 
Hence, $R$-scan must provide a randomly selected tuple of $R$ in each round so that $S$-scan observes a random reward for the current tuple of $S$. 
Because UCB-Scan sends more promising tuples more often, $R$-scan does {\it not} return random samples of tuples in $R$.
It is not clear how to devise an online learning algorithm for $S$-scan that accurately finds most rewarding tuples in $S$ without random reward samples.

\section{Online Sequential Learning}
\label{sec:singleLearning}

\begin{algorithm}
	\SetKwData{r}{r}
	\SetKwData{s}{s}
	\SetKwData{k}{k}
    \SetKwData{N}{N}
    \SetKwData{t}{t}
    \SetKwData{failure}{failure}
    \SetKwData{success}{success}
    \SetKwData{maxTuple}{maxTuple}
	\SetKwArray{rewardTable}{rewardTable}
	\SetKwData{result}{result}
	\SetKwFunction{hasNext}{HasNext}
	\SetKwFunction{readNext}{NextBlock}
	\SetKwFunction{append}{Append} 
	\SetKwFunction{join}{Join}
	\SetKwFunction{size}{Size}
	\SetKwFunction{rem}{RemoveFromDictionary}
 	\SetKwFunction{increment}{Increment}
 	\SetKwFunction{findMax}{ArgMax}
 	\SetKwFunction{reset}{Reset}
	\SetKwInOut{Input}{input}
	\SetKwInOut{Output}{output}
	\rewardTable $\gets$ $\emptyset$ \\
	\While{\size{\result} $<$ \k} { 
		\tcp{start a super round}		
		\While{\size{\rewardTable} $\le M$} {
			\tcp{Read sequentially the next tuple in $R$}
            \r $\gets$ $R$-Scan.nextTuple()\\ 
            \tcp{Estimate reward of $r$ using $N$-Failure}
            \failure $\gets$ 0\\		
            \success $\gets$ 0\\
            \tcp{Read tuples sequentially from $S$}
            \s $\gets$ $S$-Scan.nextTuple() \\
            \While{\failure $\leq$ \N} {
				\If{\r $\bowtie$ \s $\neq \emptyset$} {
                    \success++\\
                    \append{\result, \r $\bowtie$ \s}\\
                } \Else{
                    \failure++\\
                }
				\If{\size{\result} $\ge k$} {
					return\\
				}
				\s $\gets$ $S$-Scan.nextTuple() \\				
			}
            \tcp{Add information of \r to reward table}
            \rewardTable $\gets$ \rewardTable $\cup$ (\r.address, \success) \\
		}
		\tcp{Find most rewarding tuple}
        \maxTuple $\gets$ \findMax{\rewardTable} \\
		\rewardTable $\gets$ \rewardTable $\setminus$ \maxTuple \\
        \tcp{Exploit the most rewarding tuple}
        \r $\gets$ $R$-Scan.getTuple(\maxTuple.address)\\
        \append{\result, \t $\bowtie$ S} \\
	}
\caption{OSL algorithm to generate k tuples for R $\bowtie$ S} 
\label{alg:singleLearning}
\end{algorithm}
\vspace{-2pt}

In this section, we investigate learning for only one of the scan operators in the join of $R$ and $S$, namely $R$-scan.
We propose {\it Online Sequential Learning} (OSL) algorithm that addresses challenges discussed in Section~\ref{qp:challenges}.
We assume that $S$-scan has a fixed random strategy using sequential scan over $S$ as explained in Section~\ref{qp:challenges}. 

\subsection{Algorithm}

\noindent{\bf General Structure.}
Algorithm~\ref{alg:singleLearning} describes the general structure of OSL algorithm for $R$-scan. 
To simplify the exposition, we show operations of both $R$-scan and $S$-scan in Algorithm~\ref{alg:singleLearning}.
OSL constitutes a series of {\it super-rounds}.
In each super-round, our method first {\it explores and learns} the most rewarding tuple amongst all (remaining) tuples in $R$.  
It produces join results while learning.
It then {\it exploits} the information learned in the super-round and performs a full join of the selected tuple with $S$.
Our method continues to the next super-rounds until it generates sufficiently many results based on the user's preference.

\noindent{\bf Estimating Reward of a Tuple in $R$.}
As discussed in Section~\ref{qp:challenges}, due to the large number of tuples in $R$, we must estimate the reward of each tuple in $R$ quickly by joining it with as few tuples of $S$ as possible.
To do this, we leverage the exploration technique of {\it $N$-Failure}, which is used to estimate rewards of actions in infinite action spaces \cite{Berry:1997}.
To estimate the reward of a tuple in $R$, we join the tuple with a sequence of randomly selected tuples of $S$ until the tuple does not generate any result in $N$ joins.
The estimated reward for the tuple of $R$ is the total number of join results produced in these join trials. 
We pick $N$ to be significantly smaller than the number of tuples in $S$, e.g., 10 in our experiments. 
The joint tuples produced during $N$-Failure are returned to users, thus, our method generates results during learning.
Since in large databases $S$ has usually numerous tuples, one can usually estimate the reward of a large fraction of tuples in $R$ by a single sequential scan of $S$.

\noindent{\bf Exploring Sufficiently Many Tuples of $R$.}
Ideally, we would like to estimate the reward of every tuple in $R$ using $N$-Failure and pick the most rewarding one for exploitation.
However, as explained in Section~\ref{qp:challenges}, relation $R$ may contain numerous tuples.
This will cause $R$-scan to spend significant time on learning rewards before exploiting that knowledge.
Moreover, as the total reward gained by a single tuple might be limited, it might not be an efficient strategy to spend too much time on learning the tuple with the highest reward.
Thus, it might be more efficient to follow a more greedy approach in picking the most rewarding tuple in $R$ for exploitation.
To learn a set of promising tuples quickly, $R$-scan estimates the rewards of a random subset of tuples in $R$.
Since we do not assume any specific order of tuples in $R$ based on attributes involved in the join, e.g., sorted order, we may safely assume that the sequential scan of $R$ delivers random samples of its tuples \cite{Hellerstein:1997:OA:253260.253291}. 
Thus, $R$-scan sequentially scans $R$ to get a random sample of tuples from $R$.  

We have tried two method to determine when to stop exploring and estimating rewards of tuples in $R$.
In the first method, we stop exploration as soon as we find a tuple with a sufficiently large estimated reward, e.g., non-zero. 
The second method explores exactly $M$ number of randomly selected tuples in $R$ and picks the most rewarding one for exploitation.
We select $M$ to be significantly smaller than the number of tuples in $R$ but sufficiently large to be able to estimate a promising tuple.
Our initial empirical studies indicate that the latter approach delivers more efficient strategies.
Hence, we use this approach in this paper.
We show how to find a reasonable value for $M$ in Section~\ref{sec:singleLearning:analysis}.
During exploration, the reward and the address of each explore tuple of $R$ are stored in a table (reward table) in main memory.

\noindent{\bf Exploitation.}
After the exploration step in each super round, $R$-scan reads the address of the most rewarding tuple of $R$ from the reward table.
It then joins the tuple with the entire $S$, i.e., exploitation.
In this step, $S$-scan sequentially scans $S$ from the beginning.
If the reward signal is sparse, e.g., the join is very selective, in some super-rounds, the rewards of all explored tuples might be zero. 
In this case, $R$-scan picks its last scanned tuple that is already in main memory for exploitation. 

\noindent{\bf Continue Learning During Exploitation.}
As the exploration step estimates the reward of tuples of $R$ using a small subset of $S$, its estimates might not be accurate.
Thus, $R$-scan keeps updating the reward of the selected tuple during exploitation.
As soon as the updated reward of the currently exploited tuple becomes less than some other tuple in the reward table, $R$-scan replaces the tuple with the most rewarding tuple of $R$ and continues the join.
To save time, $S$-scan continues its current sequential scan of $S$. 
If needed, $S$-scan restarts its sequential scan to perform the full join with the currently exploited tuple from $R$.
To simplify the exposition of our algorithm, we do not include this feature in Algorithm~\ref{alg:singleLearning}.  

\noindent{\bf Subsequent Super Rounds.}
After each exploitation step, the reward table contains the information of $M-1$ explored tuples that have not been fully exploited.
$R$-scan randomly selects a fresh tuple of $R$ by sequentially scanning the next tuple of $R$.
It estimates the reward of the new tuple using $N$-Failure exploration technique and adds its information to the reward table.
Now, it has enough information to pick the next most rewarding tuple of $R$ for the subsequent exploitation step.

\subsection{Granularity of Actions}
Given enough available main memory, instead of using each tuple as an action, we use a sufficiently large sequence of tuples, i.e., {\bf partition}, as actions. 
This is particularly useful where the amount of reward from one tuple may not be sufficiently large to determine the most rewarding tuple(s). 
For example, in a relatively selective join, each tuple of $R$ may produce a very small number of joint tuples. 
Due to the typical estimation errors during learning, an online learning algorithm may declare several tuples with different actual rewards to be equally rewarding.
Similarly, using large partitions than a single tuple for $S$ may improve the accuracy of reward estimation for partitions in $R$ as each partition of $R$ is checked with a relatively large sample of tuples from $S$.

Using large partitions for $R$ and $S$ may also produce more results during learning and $N$-failure exploration.
As $S$-scan reads the entire $S$ once for exploiting each partition of $R$, it also significantly reduces the total number of sequential scans by $S$-scan for exploitation. 
In our reported empirical studies in Section~\ref{qp:empirical}, we set each action to be a sequence of multiple tuples for both $R$ and $S$.
Of course, if the number of tuples in each partition is too large, exploration of each partition might take too long. 
This delays exploitation and increases the running time of the algorithm. 
In our empirical studies, we treat this number as a hyper-parameter and train it over a small dataset not used for our experiment results.
In following sections, to simplify our exposition, we assume that each action is a tuple.

\subsection{Performance Analysis}
\label{sec:singleLearning:analysis}
Since we would like the learning to scale for relation $R$ with numerous tuples, it is reasonable to set $M$ to a sufficiently small value otherwise the learning of the most rewarding tuple may take significantly long time. 
On the other hand, small values of $M$ may not deliver sufficiently precise estimates. 
We theoretically analyze the proposed approach to find optimal values of $N$ and $M$.
We also use our analysis to compare OSL with similar algorithms. 

We investigate the {\it failure proportion} of the proposed method, which is the fraction of rounds during which the join does not produce any result. 
Due to the stochastic nature of our approach, we analyze its expected failure proportion.
We denote the probability of success, i.e., generating a result, for a tuple $i$ in $R$ as $p_i$.
We assume that probabilities of success of different tuples are independent and drawn from a distribution $F$.
We also assume that $F$ is a uniform distribution in the interval $[a,b]$, $0 \leq a \leq b \leq 1$. 
Using other distributions introduces substantial complexity without creating significantly different insight. 
In our empirical study in Section~\ref{qp:empirical}, we evaluate cases in which expected rewards of different tuples are drawn from distributions with various degrees of skewness. 
We denote the total number of tuples of $S$ as $\lvert S \rvert$.
The following proposition establishes a lower bound on the expected failure proportion in each super-round.
\begin{proposition}
\label{proposition:lowerBound}
The expected failure proportion in each super-round of the join has a lower bound of $(1-b) + $ $(b-a) \sqrt{\frac{2}{\lvert S \rvert}}$.
\end{proposition}
\begin{proof}
The maximum number of join operations in each super-round is $\lvert S \rvert$. 
Each success of a tuple in $R$ is equivalent to success of its joined tuple in $S$.
The result follows from Theorem 7 in \cite{Berry:1997}. 
\end{proof}

Proposition~\ref{proposition:lowerBound} holds for every learning algorithm in $R$-scan that sequentially scans and explores tuples in $R$. 
To get a clear understanding of the result of Proposition~\ref{proposition:lowerBound}, let $b=1$ and $a=0$. 
Proposition~\ref{proposition:lowerBound} indicates that a lower bound on the expected failure proportion of every learning method is $\sqrt{\frac{2}{\lvert S \rvert}}$. 
This lower bound comes from the inherent difficulty of learning the most rewarding tuple while processing the join and the restriction on the access method of $R$-scan to its actions, i.e., sequential scan.

Next, we ask whether the learning algorithm explained in the preceding section will achieve an expected failure proportion close to the aforementioned lower bound. 
The following proposition shows that it is enough to use a random sample of $R$ with a modest size to learn a tuple with an expected failure proportion close to lower bound of Proposition~\ref{proposition:lowerBound}. 
It follows from Theorem 8 in \cite{Berry:1997}.
\begin{proposition}
\label{proposition:mrun}
If $R$-scan uses OSL with $N=1$ and $M = \sqrt{\lvert S \rvert (b-a)}$, the expected failure proportion of a super-round of the join is less than or equal to $(1-b) + $ $2 \sqrt{\frac{(b-a)}{\lvert S \rvert}}$ asymptotically. 
\end{proposition}
Again, to get a better understanding of the result of Proposition~\ref{proposition:mrun}, let $b=1$ and $a=0$. 
Ignoring additive and multiplier constants, for $M=$ $\sqrt{\lvert S \rvert}$, the expected failure is at most $\frac{2}{\sqrt{\lvert S \rvert}}$.
Using this analysis, in our empirical study we set the value of $M$ to $O(\sqrt{\lvert S \rvert})$.
This analysis indicates that with a sufficiently small value for $N$ one may achieve a low expected failure proportion in each super-round.
To support other reward distributions and handle sparsity of rewards, we choose slightly larger values than $1$ for $N$ not to miss promising tuples by training over a small range of values for $N$ and set $N=10$ for the reported empirical study.

\noindent{\bf Comparison with Other Algorithms.}
Current algorithms that handle join predicates with complex conditions check each available tuple (partition) of $R$ with each available tuple (partition) of $S$ to generate join results.
For example, nested loop ripple join, iteratively reads tuples (partitions) of $R$ and $S$ to main memory and joins every tuple of $R$ with every tuple of $S$ \cite{Haas:1999:RJO:304182.304208}.
All versions of the nested loop join algorithm also join every tuple (partition) of $R$ and $S$. 
These algorithms pick each tuple (partition) of $R$ using a sequential scan and join this tuple with every tuple in $S$ via a sequential scan over $S$.
Thus, we can model these algorithms as a series of super-rounds where the most rewarding tuple of $R$ is selected randomly from $R$.
Given the aforementioned definitions for $a$ and $b$, these algorithms have the expected failure proportion of $1-\frac{a+b}{2}$, e.g.,  0.5 for $b=1$ and $a=0$, in each super-round.
For a sufficiently large $\lvert S \rvert$, our method has about $\frac{b-a}{2}$, e.g., 0.5 for $b=1$ and $a=0$, less expected failure proportion in each super-round than these methods.
This is indeed interesting as our method learns its strategy online by checking a rather small fraction of base relations.

\noindent{\bf Number of Scans for $R$ and $S$.}
During the algorithm, $R$-scan performs a single sequential scan of $R$.
$R$-scan terminates this sequential scan as soon as the desired number of joint tuples are produced. 
It may also perform some random accesses to pick the tuples with the most estimated rewards exploitation.
$S$-scan sequentially scans $S$ to provide randomly selected tuples from $S$ for $N$-Failure exploration and reward estimation.
As users need a subset of the join results in our setting and the value of $N$ is relatively small, the total sequential accesses to $S$ for $N$-Failure explorations in all super-rounds may not exceed one complete sequential scan of $S$.
$S$-scan performs a sequential scan per each exploitation (full join) in a super-round.


\begin{proposition}
Nested loop performs full join of every tuple in $R$ with entire $S$. It requires $\lvert R \rvert\lvert S \rvert$ join operations in total and $S$ is revisited $\lvert R \rvert$ times.  Assuming the tuples in $R$ are randomly distributed. Let $b=1, a=0$. In each super-round of  $M$-run exploration with $N=1$ and $M=\sqrt{\lvert S \rvert}$, the expected number of join operations is at most $\lvert S \rvert - \sqrt{\lvert S \rvert}$. If the algorithm performs $k$ super-rounds for exploration, $S$ will be revisited at most $k\frac{\lvert S \rvert - \sqrt{\lvert S \rvert}}{\lvert S \rvert}$ times. 
\end{proposition}
\begin{proof}
By Theorem 8 in \cite{Berry:1997}, the expected number of arms tried until finding a block of $R$ that joins with at least one tuple in the last $M=\sqrt{\lvert S \rvert}$ rounds is $ \sqrt{\lvert S \rvert}$. In the worst case of $S$, each of these tuples will successfully join with the block in $S$ for at most $\sqrt{\lvert S \rvert}-1$ times because otherwise it is the good tuple. Thus the total number of join operations in each super-round is at most $\sqrt{\lvert S \rvert}(\sqrt{\lvert S \rvert}-1) = \lvert S \rvert - \sqrt{\lvert S \rvert}$. By running $k$ super-rounds for exploration, $R$ will join with at most $k(\lvert S \rvert - \sqrt{\lvert S \rvert})$ blocks in $S$, resulting in $k\frac{\lvert S \rvert - \sqrt{\lvert S \rvert}}{\lvert S \rvert}$ revisits of $S$ for exploration. 
\end{proof}

We notice the benefits of our algorithm in terms of fewer scans of $S$ compared to brute force. The key insight is that brute force needs to perform full join with $S$ for every block in $R$, while our algorithm explores with fewer join operations in each super-round even in the worst case, and fewer revisits of $S$ overall. Note that in the best case, tuples in $S$ is distributed in a way that we may find the ``good'' tuple of $R$ early: the failure of the $\sqrt{\lvert S \rvert}-1$ tuples happens at the first join, followed by the success of $\sqrt{\lvert S \rvert}$ joins of the good tuple. The exploration takes only $\sqrt{\lvert S \rvert}-1 + \sqrt{\lvert S \rvert}$ join operations to find the good tuple.

\noindent{\bf Impact of Reward Distribution.}
Our analysis shows that the relative efficiency of this algorithm grows with the difference between $b$ and $a$. 
For instance, if $b=a$, as all tuples produce the same number of results, learning will not be useful no matter what method is used.
We show in Section~\ref{qp:empirical} that in cases the different tuples in $R$ have almost the same reward, our approach produces subsets of results almost as efficiently as comparable algorithms.

\section{Randomized OSL}
\label{sec:estimate}
Users often like to see accurate estimations of aggregation functions over samples of query results quickly \cite{Haas:1999:RJO:304182.304208,DBLP:conf/sigmod/ChaudhuriDK17,Agarwal:2013:BQB:2465351.2465355,Hellerstein:1997:OA:253260.253291}.
In particular, it is desirable to deliver progressively accurate results as the system produces output tuples \cite{Haas:1999:RJO:304182.304208,Hellerstein:1997:OA:253260.253291,DBLP:conf/sigmod/ChaudhuriDK17}.
Such accurate estimations usually require uniformly random samples of query results \cite{Haas:1999:RJO:304182.304208,Hellerstein:1997:OA:253260.253291,DBLP:conf/sigmod/ChaudhuriDK17}.
If the desired aggregation function is defined over attributes that do {\it not} belong to and are independent of the join attributes, OSL provides uniformly random samples.
In these cases, one may use available estimators some of which offer confidence intervals, e.g., \cite{Hellerstein:1997:OA:253260.253291}, to compute accurate estimates of the aggregation functions progressively.
Otherwise, as it {\it adaptively selects} tuples based on their rewards, current estimators may consider relatively few values from less frequently selected tuples in $R$.
As there are often too many such tuples, their estimates may be negatively biased \cite{doi:10.1073/pnas.2014602118}. 
Moreover, the observed rewards of tuples in $R$ may vary in different runs of the algorithm, e.g., due to a different order of tuples in $S$. 
This may lead to high variance in estimation. 

To address these challenges, in this section, we first propose a  modified version of OSL called {\it Randomized OSL} (ROSL),
We then design a consistent, unbiased, and asymptotically normal estimator over the results ROSL. 
Since the changes in the algorithm is minimal, the updated algorithm preserves the efficiency of our method.
Our estimator also delivers confidence intervals to inform users of how close the current estimates are to the accurate result with high probability.
This enables users to stop execution as soon as they deem the accuracy of estimate sufficient.

\subsection{Algorithm}
In this section, we modify OSL so that every tuple in $R$ has a non-zero chance of being selected for exploitation.
We first explain ROSL for $N=1$ and then describe its extension for $N > 1$.
We also explain the probabilities by which ROSL picks tuples from $R$.
These probabilities are useful for designing our estimator in Section~\ref{sec:estimate:estimator}.

\noindent{\bf Exploration Steps.}
We use the same steps of $N$-Failure explorations and reward estimations of OSL for ROSL.
As explained in Section~\ref{qp:challenges}, sequential scan provides a random sample of tuples of the underlying relation. 
Hence, during every $N$-Failure exploration, all tuples in $R$ have non-zero probability of being selected.
More precisely, the probability of choosing the first tuple $r \in R$ and trying it the first step of its $N$-Failure exploration is $\frac{1}{\mid R \mid}$.
$N$-Failure tries $r$ in each additional step with the probability of $p_r$ where $p_r$ is the probability of generating a joint tuple from $r$. 
Otherwise, it switches to another tuple $u \in R$ selected uniformly at random from $R$.
The probability of choosing $u$ in its first step of $N$-Failure exploration is $\frac{1}{\mid R \mid}$ $(1 - p_r)$.
Similarly, the algorithm picks $u$ for additional tries during its $N$-Failure exploration with the probability $p_u$.
When picking $u$, it considers only the tuples of $R$ that have not been selected so far.
To ensure that every tuple has a non-zero probability of selection, one may also consider the previously accessed tuples when choosing $u$.
Nonetheless, because our goal is to generate a sample of join over large relations, the number of previously accessed tuples is usually significantly less than the ones that have not been tried. 
Thus, we do not change OSL selection method in this step.

\noindent{\bf Exploitation Steps.}
After finishing $N$-Failure explorations in each super-round $\Gamma$, OSL picks the highest rewarding tuple $r^*_{\Gamma}$ to join with $S$, i.e., {\it exploitation phase} of $\Gamma$.
This method may induced some bias toward $r^*_{\Gamma}$.
To mitigate this bias, in each step of the exploitation phase of $\Gamma$, ROSL picks tuples of $R$ that are explored in $\Gamma$ and its preceding super-rounds randomly proportional to their reward.
We set the rewards of tuples with zero reward to a fixed and relatively small value so they are selected with non-zero probability.
This exploitation phase in each super-round ignores tuples of $R$ that have not yet been explored.
Nonetheless, since it picks the explored tuples in each super-round uniformly at random, this method does not introduce any bias, e.g., similar to joins of randomly selected groups of tuples in ripple join \cite{Haas:1999:RJO:304182.304208,Hellerstein:1997:OA:253260.253291}.
As in OSL, we update the reward of the selected tuple during exploitation. If the reward of the selected tuple is less than the reward of some explored tuples in $R$, we randomly select another explored tuple from the reward table and continue exploitation.

\noindent{\bf Larger Values of $N$.}
Assume that $N > 1$.
The $M$-Run algorithm joins each selected tuple in $R$ with $S$, $N$ times at the start of each $N$-Failure exploration.
Other tuples in $R$ have zero probability of being selected during these $N$ steps.
Nonetheless, these $N$ joins do not introduce any bias to the sample.
First, each tuple for these trials is selected uniformly at random from $R$.
Second, they are performed for all these tuples equal number of times. 
They are equivalent to trying each tuple $N$ times before adaptive selections begin. 
Hence, we leave these parts of our algorithm unchanged.
Following the calculation in the preceding paragraph, the probability of trying tuple $r$ after the first $N$ joins in its $N$-Failure exploration is $1 - (1 - p_r)^{N}$.
Similarly, the probability of switching to a new tuple $u$ is $\frac{1}{\mid R \mid}$ $\times (1 - p_r)^{N}$.
Since value of $p_r$ are not known beforehand, we use the empirical observation from the join of $r$ with tuples in $S$ during $N$-Failure exploration to estimate $p_r$.

\subsection{Unbiased and Consistent Estimator}
\label{sec:estimate:estimator}
In this section, we design a consistent, unbiased, and asymptotically normal estimator for an aggregation function $Q(.)$, e.g., count or sum,  over the results of ROSL for $R \bowtie S$. 

\noindent{\bf Unbiased Mean Estimator for Each Tuple.}
We first design an unbiased estimator for the mean of $Q(.)$ computed only by join of tuple $r$ in $R$ with $S$, denoted as $Q(r)$.
As explained in the preceding section, during both exploration and exploitation, ROSL may pick different tuples of $R$ with different probabilities at each step.
This may introduce bias to the sample of join results produced by ROSL.
We use the popular approach of inverse-selection-probability-weighting (ISPW) to correct that bias \cite{Horvitz1952AGO}.
Let ROSL picks $r$ at step $t$ of the algorithm with probability $e_t(r)$.
Let $Y_t(r)$ denote the value of the function from the joint tuple created by joining $r$ with a tuple in $S$ in step $t$ of ROSL.
Let $T$ be the fixed times when we compute and deliver updated estimates of the function to the user, e.g., every 100 steps.
The {\it history} of the algorithm up to time $t-1$, shown as $H^{t-1}$, is the sequence of joint tuples and their function values up to time $t-1$.
We define $\hat{\Gamma}_t(r)$ as $\frac{Y_t(r)}{e_t(r)}$ and estimate $Q(r)$ at time $T$ using $\hat{Q}_T(r) = $ $ \frac{1}{T} \sum_{t=1}^{T} \hat{\Gamma}_t(r)$.
Because given $r$ the samples of $Y_t(r)$ are independent, due to the linearity of expectation, $\hat{Q}_T(r)$ is unbiased. 


\noindent{\bf Asymptotically Normal Mean Estimator for Each Tuple.}
Using ISPW technique to estimate $\hat{Q}_T(r)$ may not deliver a deterministic and controlled variance and achieve asymptotic normality \cite{doi:10.1073/pnas.2014602118}.
To address this issue, we use adaptive weights and modify $\hat{Q}_T(r)$ as $\frac{ \sum_{t=1}^{T} h_t(r) \hat{\Gamma}_t(r) }{\sum_{t=1}^{T} h_t(r)}$ where $h_t(r) = $ $\sqrt{\frac{e_t(r)}{T}}$ following \cite{doi:10.1073/pnas.2014602118}.
The weights $h_t(r)$ compensate for uncontrolled trajectories of $e_t(r)$ and help to stabilize variance. 
Since $\sum_{t=1}^{T} \frac{{h_t^2(r)}}{{e_t}}$ $=1$, $\hat{Q}_T(r)$ is still
unbiased.
The following is a direct result of Theorem 2 in \cite{doi:10.1073/pnas.2014602118}.
\begin{proposition}
\label{proposition:meanestimate:onetuple}
$\hat{Q}_T(r)$ is consistent and $\frac{\hat{Q}_T(r) - Q(r)} {\hat{V}_T(r)^{1/2}}$ distributed asymptotically normally with mean of 0 and variance of 1 where
$\hat{V}_T(r)$ $=$ $ \frac{\sum_{t=1}^{T} h_t^2(r) \left(\hat{\Gamma}_t(r) - \hat{Q}_T(r)\right)^2}{\left( \sum_{t=1}^{T} h_t(r) \right)^2}$.
\end{proposition}
\begin{proof}
The samples of $Y_t(r)$ are independent given $r$.
To satisfy the conditions of Theorem 3 in \cite{doi:10.1073/pnas.2014602118}, we show that each tuple has at least a fixed non-zero probability of being selected and the probability of choosing for every tuple converges to a value in $[0,1]$.
First, we have $e_t(r) >0$.
Second, each tuple $r$ is explored in some super-round.
$e_t(r)$ in each examined super-round converges to the probability computed in ROSL based to its reward by time $T$.
\end{proof}

\noindent{\bf Mean Estimator for The Function.}
Let $Q$ be the mean of the aggregation function computed using all tuples.
Let $T_r$ denote the number of times the algorithm tries and joins a tuple $r$ during $T$ steps of the algorithm. 
We estimate the mean as $\hat{Q}_T =$ $\frac{\sum_r T_r \hat{Q}_T(r)}{T}$. 
\begin{theorem}
\label{theorem:asymtoticnormal}
$\hat{Q}_T$ is consistent and $\frac{\hat{Q}_T - Q} {\frac{\sum_{r} T_r \hat{V}_T(r)^{1/2}}{T}}$ is distributed asymptotically normally with mean of 0 and variance of 1.
\end{theorem}
\begin{proof}
In the first super-round, the algorithm picks a set tuples from $R$ uniformly at random using sequential scan.
Assume that the algorithm is in the first super-round.
Using Proposition~\ref{proposition:meanestimate:onetuple}, for tuple $r$ from $R$ in the current super-round, $\hat{Q}_T(r) - Q(r)$ converges in distribution to $\mathcal{N}(0, \hat{V}_T(r)^{1/2})$.
Because multiplying by a fixed value is a continuous function, $W(r) = T_r$ $\hat{Q}_T(r) - Q(r)$ converge to a normal with mean $0$ and variance $T_r \hat{V}_T(r)^{1/2})$.
As the aggregation function values of tuples are independent and average is a continuous function, by applying Portmanteau theorem, the average of all $W(.)$ values for the tuples of $R$ in the current super-round is distributed asymptotically normally with mean $0$ and variance $\frac{\sum_{r} n(r) \hat{V}_T(r)^{1/2}}{T}$.
The algorithm moves to a subsequent super-round by randomly selecting tuples from $R$.
If it picks tuples from the first super-round, the tuple will be picked randomly. 
Thus, the convergence holds for subsequent steps of the algorithm.
Because each tuple of $R$ joins with at least $N$ tuples of $S$, the selected random samples by the algorithm from $R \bowtie S$ are not independent.
Thus, similar to \cite{Haas:1999:RJO:304182.304208,Hellerstein:1997:OA:253260.253291}, we use the central limit theorem for dependent samples for the above convergence results. 
\end{proof}
\noindent
As variance vanishes when $T \rightarrow \infty$, the estimator is consistent, i.e., gets arbitrary close to the function mean given enough samples.
We use the results of Theorem~\ref{theorem:asymtoticnormal} and construct $p$ confidence interval $P(\hat{\mu} - \epsilon_n < \mu < \hat{\mu} + \epsilon_n)$ $= p$ where $\epsilon_n =$ $z_p$ $\sum_{r} \hat{V}_T(r)^{1/2}$ and the area under the standard normal curve between $-z_p$ and $z_p$ is $p$.
We set $p=0.95$ in our experiments.

\noindent{\bf Estimating The Function Value.}
One may estimate the value of the function from its mean estimate by normalizing the current estimate \cite{Haas:1999:RJO:304182.304208}.
For instance, to estimate {\it count}, we multiply its estimated mean by $\frac{\mid R\mid \mid S \mid}{\mid J \mid}$ where $J$ is the number of current sample.
\section{Collaborative Scans}
\label{sec:collaboration}
\noindent{\bf Benefit of Collaborative Learning.}
As explained in Section~\ref{qp:challenges}, in $m$ to $n$ joins of $R$ and $S$, some tuples of $S$ may join with a considerable number of tuples in $R$.
Thus, it may improve the efficiency of the join if $R$-scan and $S$-scan both use adaptive strategies and {\it collaboratively learn} an efficient join strategy.

\noindent{\bf Challenge of Collaborative Learning.}
In collaborative learning, \ul{each scan must both learn the reward of its tuples and provide randomly sampled tuples for the other scan so they both learn accurate rewards}. 
This may double the number of data or I/O accesses by each scan during learning and slow down join processing significantly.
For instance, one approach is to run the join in two main phases. 
In the first one, $R$-scan uses a learning and $S$-scan uses a fixed random strategy as in Section~\ref{sec:singleLearning}. 
After exploring and exploiting all tuples of $R$, in the second phase, $S$-scan follows a learning strategy and $R$-scan a fixed random one.
In this method, the learning and data accesses in the second phase are not useful as all join results have been already produced by the first phase.
Hence, this simple algorithm does not use highly rewarding tuples in $S$ to generate results quickly.

\subsection{Interleaving Learning and Sampling}
\label{sec:collaboration:inter}
We propose a method to address the aforementioned challenge so both scans learn join strategies efficiently.
OSL (ROSL) needs a random subset of actions for exploration. 
Thus, the scan operator that follows this algorithm reads its tuples from its relation sequentially.
Also, the scan operator with the fixed strategy scans its relation sequentially to provide random samples.
Since both scans sequentially read their relations, we can conveniently {\it interleave} exploration/learning and fixed random sampling strategies for each scan.
This way, in each super-round, one scan performs an $N$-Failure exploration and the other one performs a sequential scan.
They will switch strategies in the subsequent super-round. 

More precisely, in a super-round, the scan that performs $N$-Failure estimates and collects rewards of its explored tuples. 
After exploring the last tuple in the super-round, this scan decides its most rewarding tuple and fully joins with the other relation.
This scan will follow a fixed random strategy in the next super-round by continuing its sequential scan of its relation.
The scan with the fixed strategy in this super-round will perform $N$-Failure exploration and subsequently exploration in the next super-round.
The join will stop after generating the desired number of tuples.
The join operator avoids reporting duplicates using techniques in Section~\ref{qp:challenges}.

\noindent{\bf Estimating Functions.} Both scans learn rewards for their tuples using random samples of the data and follow OSL (ROSL).
Hence, if each scan uses ROSL, the guarantees proved for the estimator in Section~\ref{sec:estimate} will extend for collaborative learning.
In this case, the estimator will use the results from both scans. 

\subsection{Implicit Collaborative Learning} 
\label{sec:collaboration:implicit}
The interleaving approach requires a distinct set of explorations to learn rewarding tuples in each relation.
It doubles the amount of exploration compared to the case where only a single scan learns join strategies.
Next, we propose a method that learns the rewards of tuples in both relations by performing exploration only on one of the relations.
This method spends as much time as OSL does on exploration and learning.
In this method, $R$-scan uses OSL and $S$-scan follows a fixed random strategy.
The main idea of this approach is for $S$-scan to collect and use information of joins performed during $N$-Failure explorations of $R$-scan to estimate the rewards of tuples in $S$.
$S$-scan uses this information to simulate $N$-Failure exploration for its tuples without doing any additional joins.

More precisely, $S$-scan collects the information on the results of joins of each tuple in $S$ with $N$ ($P$) randomly selected tuples of $R$ in which the joins fail (succeed) 
as required by $N$-Failure algorithm.  
As soon as $S$-scan has sufficient information to estimate the most rewarding tuple in $S$, it exploits this tuple by joining with with the entire $R$.
In the meantime, since $R$-scan follows OSL, it may collect enough information to exploit multiple tuples.
Hence, as opposed to the interleaving approach, $S$-scan may exploit its information after more than a single exploitation of $R$-scan.

\noindent{\bf Collecting Information from Random Samples of $R$.}
To simulate $N$-Failure for its tuples, $S$-scan needs information of joins between these tuples and randomly selected tuples from $R$.
Nonetheless, since $R$-scan uses $N$-Failure during its exploration, it may keep using the same tuple of $R$ for join in some consecutive join trials.
Thus, it may not provide random samples of $R$ tuples.
To estimate the reward of a tuple in $R$, $R$-scan initially joins the tuple $N$ times during its $N$-Failure exploration regardless of its reward.
$R$-scan also randomly selects each tuple in $R$ at the start if its $N$-Failure exploration.
Hence, for the tuples of $S$ that participate in these $N$ joins, we consider these joins to be joins with a randomly selected tuple from $R$ and use the information from these joins to estimate the rewards of these tuples in $S$.
$S$-scan does not use the information of results of joins during the exploitation of tuples in $R$ as the tuple chosen from $R$ for exploitation is not selected uniformly at random.

\noindent{\bf Estimating the Most Rewarding Tuple in $S$.}
In OSL, while $R$-scan performs $N$-Failure explorations, $S$-scan sequentially scans $S$ to provide $R$-scan with fresh randomly selected tuples. 
If $S$ contains many tuples, using this method, $S$-scan may not collect sufficient information on each tuple in $S$.
Thus, we limit the total number of tuples in $S$ used for $N$-Failure explorations of $R$-scan.
We set this value to $\sqrt{|R|}$ due to the analysis in Section~\ref{sec:singleLearning:analysis}.
When $S$-scan exhaust these tuples during $N$-Failure explorations of $R$-scan, it restarts its scan from the beginning of $S$.

\noindent{\bf Exploitation.}
As soon as $S$-scan gather sufficient information on to estimate rewards of $\sqrt{|R|}$ tuples in $S$ using $N$-Failure, it picks the most rewarding tuple in $S$ and joins it with 
the entire $R$. 
After this exploitation step, $R$-scan resumes its $N$-Failure explorations or a round of exploitation depending on its available information.

\noindent{\bf Subsequent Implicit Exploration in $S$.}
Since $S$-scan has reward estimation for $\sqrt{|R|}$-1 tuples (excluding the exploited tuple) of $S$, it needs to gather sufficient information on an additional tuple of $S$ to decide the most rewarding tuple for its next exploitation step.
Hence, $S$-scan sequentially reads and uses $L$ tuples that appear next to the first $\sqrt{|R|}$ tuples in $S$ in the subsequent $N$-Failure explorations of $R$-scan.
The value of $L$ must be sufficient large to enable $R$-scan to estimate the reward of its tuples during $N$-Failure explorations, i.e., $L > N$.
Our initial empirical studies suggest that a value of $L = 2N$ tuples in $S$ is sufficiently large to be used for $N$-Failure explorations of tuples in $R$.
As soon as $S$-scan gathers enough information to simulate $N$-Failure exploration on one of these $L$ tuples, it will have enough information to find the most rewarding tuple in $S$.
Hence, it performs another round of exploitation.

\vspace{-4pt}
\subsection{Joining More than Two Relations}
Our framework in Section~\ref{qp:framework} is based and generalizes the query operator framework in relational data systems.
Thus, our proposed algorithms can be used in an operator tree with other join methods.
For instance, to execute a left-deep operator join tree of three relations, one may use collaborative learning for the leaf join and nested loop for the root one.
One might also use OSL for only one scan of the leaf join and the scan of the root join.

Generally speaking, it is useful for an operator to use a learning strategy: 1) if the operator must choose a subset of a relatively large set of tuples for join, and 2) the exploitation steps in all super-rounds have sufficiently many possible joins whose total rewards justifies the time spent on exploring and learning an optimal strategy. 
A learning-based strategy might not be useful in other settings. 

For instance, consider a left-deep join tree with three base relations where both scans of the leaf join collaboratively learn a join strategy as explained in this section.
A database system might use a pipelining approach in executing this query and performs a join between a joint tuple of the leaf join $t$ and all tuples in the third relation as soon as $t$ produced.
Because in this example, there is only one possible strategy for join, there is not any need for a learning algorithm to find the most efficient strategy.

Query optimizers often decide which algorithms to use for operators in a query plan.
Thus, a query optimizer may use this information and assign proper algorithms to operators in the aforementioned example.
An interesting future work is to investigate the changes in the query optimizer to consider our query processing methods. 

\section{Empirical Study}
\label{qp:empirical}
\subsection{Experimental Setting}
\label{qp:empirical:setting}
\begin{table}[h]
\caption{Information on TPC-H Relations}
\vspace{-8pt}
\begin{tabular}{|c|c|c|}
\hline
\textbf{Relation} & \textbf{Primary Key} & \textbf{Size (\# Tuples)} \\
\hline\hline
customer & custkey & 7,500,000 \\ 
\hline
supplier & suppkey & 500,000 \\ 
\hline
part & partkey & 10,000,000 \\ 
\hline
partsupp & (partkey, suppkey) & 40,000,000 \\ 
\hline
orders & orderkey & 75,000,000 \\ 
\hline
lineitem & (linenumber, orderkey) & 300,000,002 \\ 
\hline
\end{tabular}
\label{tab:tpch-table}
\vspace{-5pt}
\end{table}
\noindent{\bf Methods.}
We evaluate our proposed methods and the comparable methods explained in Section~\ref{sec:intro} that do not require lengthy preprocessing to create (large) auxiliary data structures, e.g., index, and are not limited to certain types of joins, e.g., only equijoin.
We compare our methods to {\it nested loop ripple join} \cite{Haas:1999:RJO:304182.304208}, {\it nested loop} (NL), and {\it block nested loop} (BNL).
The nested loop ripple join quickly runs out of memory in all settings before generating answers as explained in Section~\ref{sec:intro}.
Previous work has found similar results about ripple join \cite{DBLP:conf/sigmod/JermaineDAJP05}.
As progressive sort-based and hash-based methods support a limited type of joins, i.e., equi-joins, they are not comparable to ours.  
But, to get a better understanding of the performance of our methods, we compare them to {\it Sort-Merge-Shrink} (SMS) join \cite{DBLP:conf/sigmod/JermaineDAJP05}.
SMS is a sort-based approach to generate a subset of results of equi-joins progressively and quickly.
Empirical studies indicate that SMS outperforms hash-based algorithms \cite{DBLP:conf/sigmod/JermaineDAJP05}.
{\it UCB-Scan} did not return any result for our queries after 15 minutes so we do not report its results.

\noindent{\bf Data.}
We use TPC-H benchmark \cite{TPC} to generate the queries and databases for our experiments. 
TPC-H is a benchmark for decision support queries that usually contain many join operators and is widely used to evaluate join processing techniques. 
We generate a TPC-H database of size 50 GB.
Table~ \ref{tab:tpch-table} contains information of TPC-H relations used in our experiments. 
One of the parameters that impacts the join processing time is the distribution of the join attribute values; more specifically their skewness. 
It has been recognized that data skew is prevalent in databases and real world data distributions are often non-uniform \cite{elseidy2014scalable,DBLP:conf/sigmod/SalzaT89}. 
For instance, different customers usually place significantly different number of orders and different orders may contain various number of line items.
We evaluate the query run-time over different versions of TPC-H database with different degrees of skewness in the join attributes. 
The Zipfian distribution (Zipf) has been used to evaluate the performance of query execution methods on skewed data distributions \cite{elseidy2014scalable,DBLP:conf/sigmod/SalzaT89,christodoulakis1983estimating}. 
Our experiments use Zipf distributions with $z$ values ranging between $[0, 0.5, 1, 1.5]$ wherein $z = 0$ will result in uniform distribution and for $z \geq 0.5$, the distribution becomes more skewed as the value of $z$ grows. 
The data distribution $z = 0.5$ is nearly uniform and the results of our experiments using $z = 0.5$ are similar to the ones with $z = 0$, therefore, we do not report results of $z = 0.5$.

\noindent{\bf Queries.} 
We extract a set of 1-N join queries from TPC-H workload.
The original TPC-H query numbers and their corresponding (extracted) 1-N joins are:  
$Q_{12}:$ $orders \bowtie_{orderkey} lineitem$, $Q_{10}:$ $customer \bowtie_{custkey} orders$, $Q_{15}:$ $supplier \bowtie_{suppkey} lineitem$.
We extract two M-N joins from TPC-H workload: 
$Q_{9}:$ $partsupp \bowtie_{partkey} lineitem$,  $Q_{11}:$ $orders \bowtie_{o\_orderdate=l\_shipdate} lineitem$. 
To evaluate the performance of our methods over joins with complex conditions, we modify the join predicates of these queries to resemble {\it similarity joins}.
Due to mistakes in data entry, data values often contain errors \cite{chu2016data}.
For instance, foreign keys in a relation might not precisely match the primary keys in another relation due to errors in values of foreign keys. 
We modify the join predicate of TPC-H queries and use {\it Levenshtein distance} instead of equality to implement approximate match between join attributes.
It allows for edit distance of at most one between join attributes.
We run each query 3 times by randomly shuffling orders of tuples and report the average. 
The reported times include both {\it learning} and {\it execution}.
\noindent{\bf Platform.}
We have implemented our proposed and other methods inside PostgreSQL 11.5.
We have performed our empirical study on a Linux server with Intel(R) Xeon(R) 2.30GHz cores and 500GB of memory. 
The size of available memory for each query on PostgreSQL is set to the minimum possible value of $0.5$ GB to simulate available memory for a single query in a server. 
SMS did not return any result after 15 minutes except for $Q_{12}$. 
To get more insight to the relative performance of SMS and our proposed methods, we report the results of SMS using $4$ GB memory for other queries.

\noindent{\bf Hyper-Parameters.}
We have used query $Q_{12}$ to train the value of $N$ and set $N=10$ for all experiments.
We set the value of $M$ to $(\sqrt{\lvert S \lvert})$ following the analysis in Section~\ref{sec:singleLearning:analysis}.
We train the size of each partition using $Q_{12}$ and set each partition to 320 tuples.

\subsection{Experimental Results}

\subsubsection{Single Scan Learning} 
\label{sec:empirical:single}
We first analyze the performance of OSL in which only one of the scan operators learn.
Considering a 1-N join query, it is reasonable for the scan operator over the relation whose every tuple may join with multiple tuples of the other relation to adopt a learning strategy.
To compare our algorithm to NL (BNL), we use the relation with learning scan as the outer relation in NL (BNL).
To simplify our description, we refer to this relation as outer relation in all methods.
For each M-N join query, we randomly pick one of the relation to be the outer relation for OSL and NL (BNL).

Figure~\ref{fig:sl-k1000} illustrates the discounted weighted average time of OSL (solid curve) and NL (dashed curve) for equijoin queries over different output sample sizes ($K$) in log-log graphs. 
Since the results of BNL are similar to NL, we do not show them in Figure~\ref{fig:sl-k1000}.
Since the queries $Q_{9}$, $Q_{10}$, $Q_{11}$, $Q_{12}$, $Q_{14}$, $Q_{15}$ deal with relations of different sizes, their full joins generate different total number of tuples. 
Thus, we use a different maximum values for $K$ per query based on their total output. 
Each graph compares the performance of both the algorithms as the $z$ value increases from 0 to 1.5. 
Each color represents the results using a different $z$ value. 
As NL (BNL) has not produced any result after the maximum time limit of 15 minutes in some settings, its results are not shown in some graphs. 
OSL significantly outperforms NL (BNL) for all queries and all values for $z$.
Generally, the more skewed the data is, the more improvement OSL shows compared to NL (BNL). 

Figure~\ref{fig:o_sms} shows that performance of SMS joins over equijoin queries using the same setting as the one in Figure~\ref{fig:sl-k1000}.
OSL considerably outperforms SMS for all queries and datasets.
As explained in Section~\ref{qp:empirical:setting}, SMS does not produce any result using 0.5GB of main memory for any query and dataset after 15 minutes. 
Hence, we use 4GB for SMS in the results reported in Figure~\ref{fig:o_sms}.
We use the default memory configuration of 0.5GB for OSL for its results in Figure~\ref{fig:o_sms}.
Although OSL is not built to support only joins with equality predicates, it outperforms SMS that is developed for them.


\noindent{\bf Accuracy of Estimation.}
Figure~\ref{fig:mrun-q10} shows the average mean estimation error rates for the aggregation query COUNT on Q10, using ROSL and the estimation technique explained in the section~\ref{sec:estimate:estimator}. 
We use $z$ values $0$, $1$, and $1.5$. 
The number of join results for all values of $z$ is about 75,000,000.
The plot shows a relatively steep decline in error rates before 1000 output samples, which is about 0.001\% of the total results.  
For large values of $z>0$, the estimator returns about 90\% accuracy using only close to 1000 samples.
For $z=0$, the estimator requires a larger but still modest sample size of 3000 (about 0.004\% of the actual result) to deliver 90\% accuracy.

\begin{figure*}
  \centering
  \includegraphics[height=40mm]{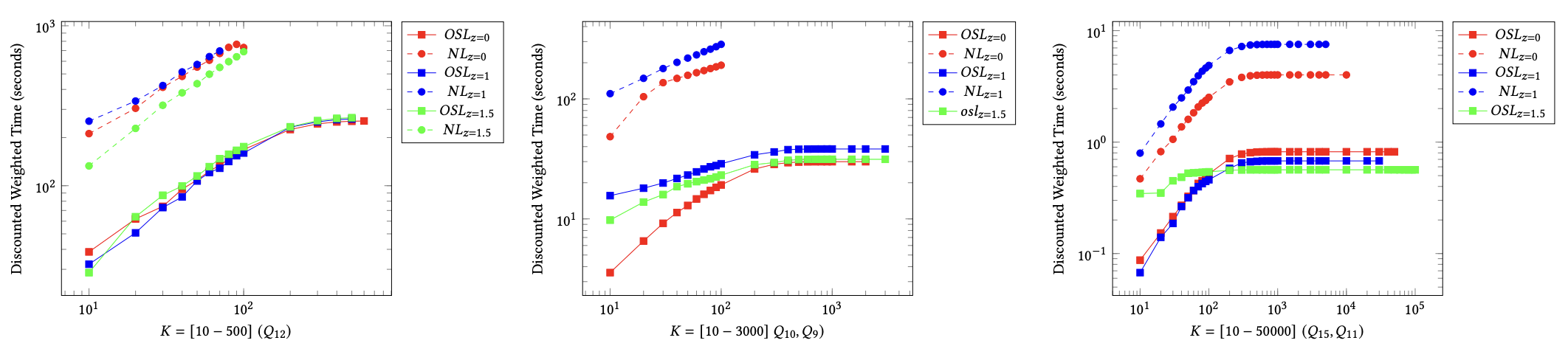}
  \caption{Average Response times for OSL and NL over comparable queries.}
   \label{fig:sl-k1000}
\end{figure*}

\begin{figure*}
  \centering
  \includegraphics[height=40mm]{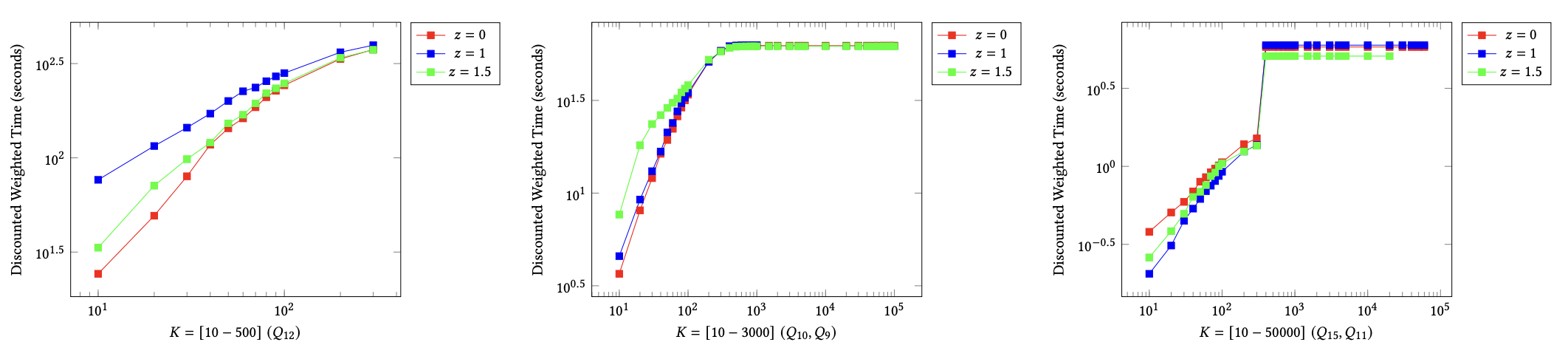}
  \caption{Discounted weighted average times for SMS over comparable queries.}
 \label{fig:o_sms}
\end{figure*}

\begin{figure*}
  \centering
  \includegraphics[height=40mm]{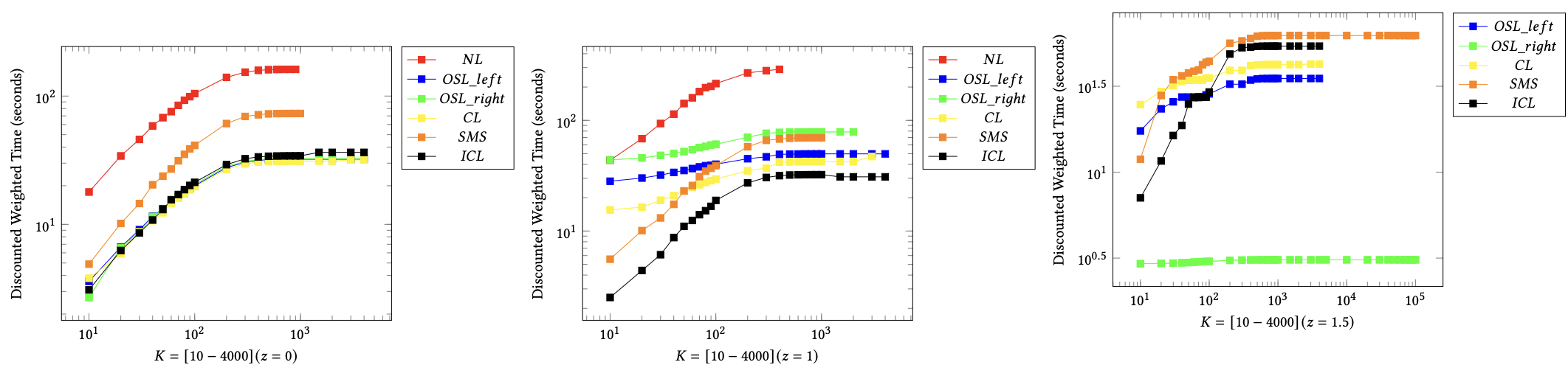}
  \caption{Discounted weighted average times for CL, ICL,  OSL\_right and OSL\_left, and NL over $Q_{9}$.}
 \label{fig:cl-k5000}
\end{figure*}

\begin{figure*}
  \centering
  \includegraphics[height=40mm]{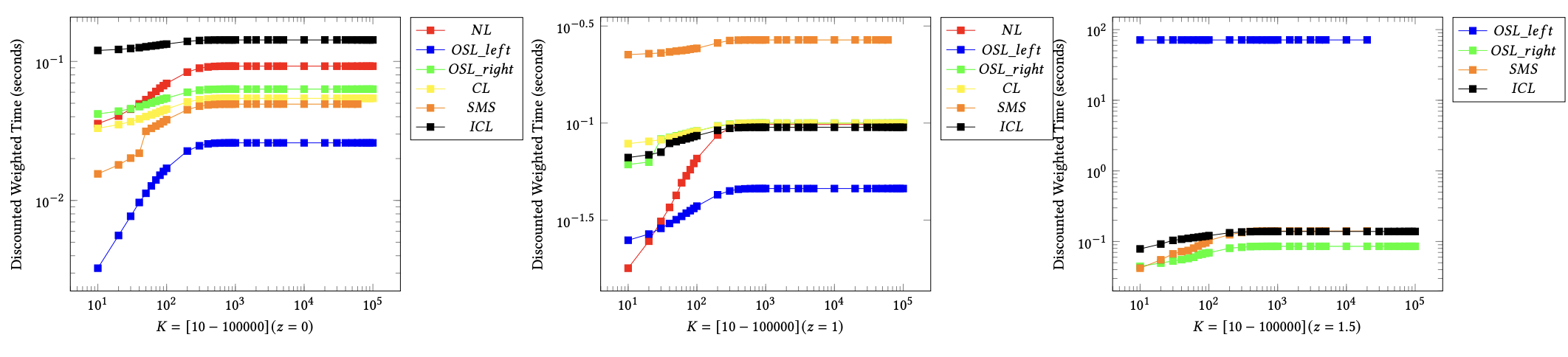}
  \caption{Discounted weighted average times for CL, ICL, OSL\_right and OSL\_left, SMS, and NL over $Q_{11}$.}
 \label{fig:cl-k50000}
\end{figure*}

\begin{figure*}
  \centering
  \includegraphics[height=40mm]{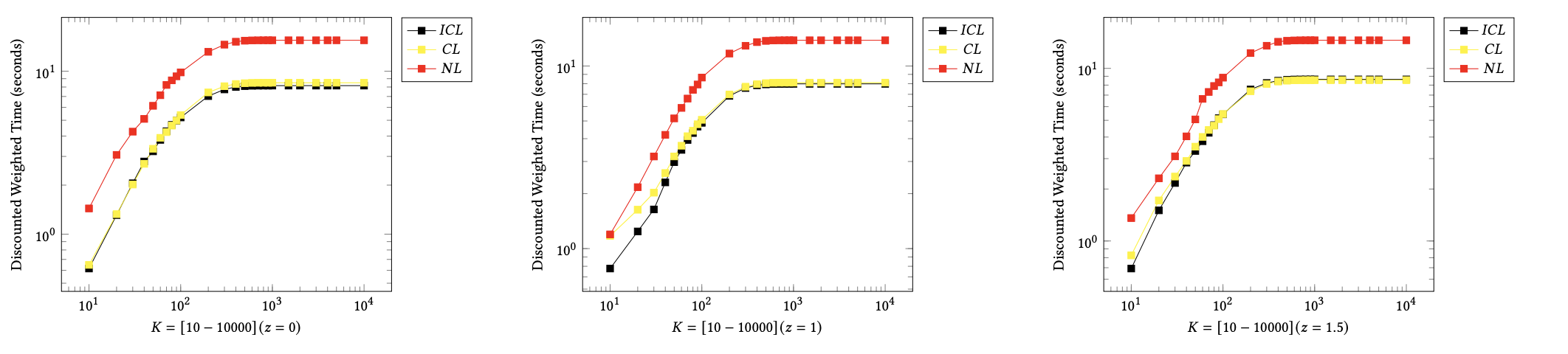}
  \caption{Discounted weighted average times for CL, ICL, and NL over $Q_{10}$ with non-equality predicate.} 
 \label{fig:like-cl-k5000}
\end{figure*}

\begin{figure*}
  \centering
  \includegraphics[height=40mm]{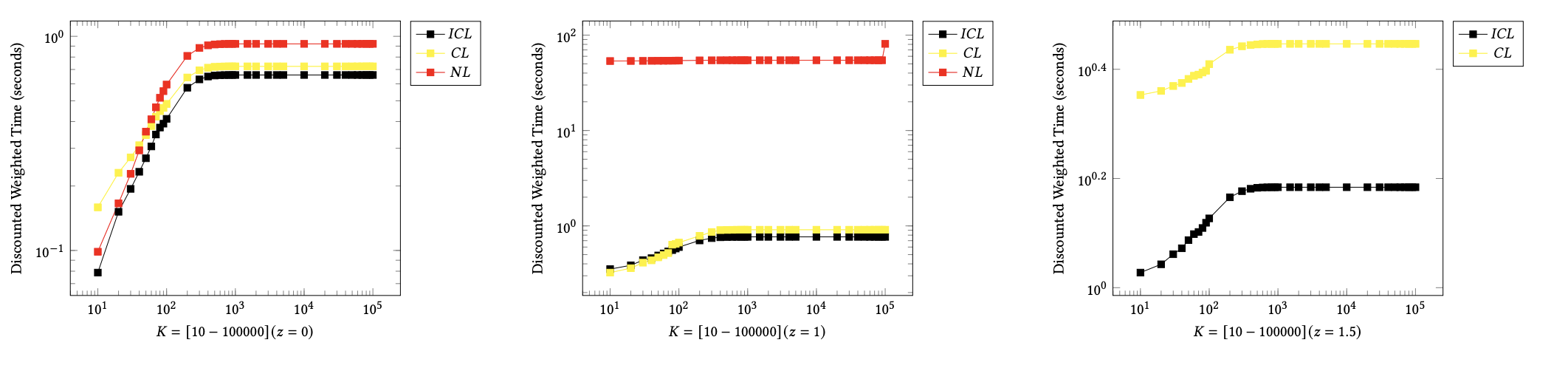}
  \caption{Discounted weighted average times for CL, ICL, and NL over $Q_{15}$ with non-equality predicate.} 
 \label{fig:like-cl-k5000-Q12}
\end{figure*}

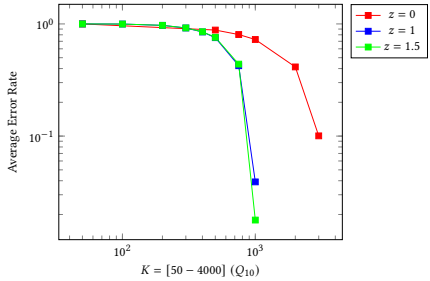
\begin{figure}
\centering
\begin{subfigure}{.33\textwidth}
        \begin{tikzpicture}[scale=0.55]
        \begin{loglogaxis}[
        	xlabel={$K = [50 - 4000]$ $(Q_{10})$},
        	ylabel={Average Error Rate},
        	legend pos=outer north east
        ]
        \addplot [mark=square*,mark options={scale=1,solid}, red]
        coordinates {
            (50,	0.999999475	)
            (500,	0.881749101	)
            (750,   0.804223403 )
            (1000,	0.726440099	)
            (2000,	0.413686219	)
            (3000,	0.100484894	)
        };
        \addlegendentry{$z=0$}
        \addplot [mark=square*,mark options={scale=1,solid}, blue]
        coordinates {
            (50,	0.999999697	)
            (100,	0.997492451	)
            (200,	0.970973954	)
            (300,	0.920164289	)
            (400,	0.848004376	)
            (500,	0.754297982	)
            (750,   0.424128865 )
            (1000,	0.039062822	)
            
        };
        \addlegendentry{$z=1$}
        \addplot [mark=square*,mark options={scale=1,solid}, green]
        coordinates {
            (50,	0.999999088	)
            (100,	0.99733635	)
            (200,	0.971406374	)
            (300,	0.921298088	)
            (400,	0.851299829	)
            (500,	0.759818424	)
            (750,   0.438124806 )
            (1000,	0.017840885	)
        };
        \addlegendentry{$z=1.5$}
        \end{loglogaxis}
        \end{tikzpicture}
 \end{subfigure}
\caption{Average error rates for query $Q_{10}$ using ROSL.}
\label{fig:mrun-q10}
\end{figure}

\subsubsection{Collaborative Learning} 
\label{sec:empirical:co}

\noindent{\bf Equality Joins.}
We compare the discounted weighted average times of the collaborative learning (CL) proposed in Section~\ref{sec:collaboration:inter}, implicit collaborative learning (ICL) proposed in Section~\ref{sec:collaboration:implicit}, OSL, NL, BNL, and SMS for different output samples sizes (K) using two M-N equi-join queries $Q_{9}$ and $Q_{11}$. 
We report two runs of OSL for each query in each of which one of the scans uses OSL algorithm (denoted as {\it OSL\_left} and {\it OSL\_right}).
Figures~\ref{fig:cl-k5000} and \ref{fig:cl-k50000} show the response times of CL, ICL, OSL\_left, OSL\_right, NL, and SMS over $Q_{9}$ and $Q_{11}$, respectively, using log-log graphs. 
Since results of BNL are similar to NL, we do not show them in these figures.
NL (BNL) do not produce any results over the dataset with $z=1.5$.

ICL in all settings but $Q_{11}$ and $z=0$ delivers the same performance or improves the performance of CL. 
In each setting, one of the relations may contain more tuples with high rewards that the other relation. 
Therefore, using OSL for one of the scan operators may be more efficient than using CL or ICL, e.g., OSL\_left for $Q_11$ and $z=0$ and $z=1$ or OSL\_right for $Q_11$ and $z=1.5$.
Nonetheless, it is {\it not} clear which one of scan operators should use a learning strategy to deliver the most efficient results before running the query.
As observed in Figures~\ref{fig:cl-k5000} and \ref{fig:cl-k50000} ICL and CL provide a middle-ground between these two possible configurations for using OSL.

Learning-based methods in all settings but ICL in $z=0$ and $Q_{11}$ significantly outperform NL (BNL).
The differences between the learning algorithms and BNL become generally more significant as the skewness parameter $z$ increases. 
ICL also generally outperform or deliver the same performance as SMS in all settings but $Q_{11}$ and $z=0$.
The reported results for SMS use significantly larger memory (4GB) than other methods (0.5GB).

\noindent{\bf Non-Equality Joins.}
Figures~\ref{fig:like-cl-k5000} and \ref{fig:like-cl-k5000-Q12} reports the results CL, ICL, and NL over non-equi join query created from $Q11$ and $Q15$. respectively, using log-log graphs. 
As the results of BNL are similar to NL, we do not report them in these graphs.
We observe almost a similar trend to the ones of equi-join queries where learning-based methods substantially outperform NL (BNL).
This difference is more significant for values of $z > 0$ in the results of $Q_{15}$ as due to the skewness in the data learning-based methods may find more highly rewarding partitions.
NL has not produced any answers for $Q_{15}$ and $z=1.5$ after about 15 minutes.
Similar to the results of M-N equi-join queries, ICL delivers the same or better performance than CL.

\section{Related Work}
\label{qp:related}
{\bf Adaptive Query Processing.} 
Researchers have proposed adaptive query optimization and planning techniques to deal with the lack of information about the query workload \cite{elseidy2014scalable,kaftan2018cuttlefish}.
As opposed to our framework, this line of research is about finding the most efficient algorithm(s) from a fixed set of traditional query execution algorithms, e.g., traditional join algorithms, to run operators in a query plan.
For instance, researchers in \cite{elseidy2014scalable} use an especial operator in the query plan that chooses the best algorithm between using nested loop, sort-merge, or hash-join to execute its child join operator.
In our approach, however, each query operator aims at learning an efficient algorithm of accessing tuples of underlying relations and generating results.
Also, our methods deal with a significantly larger set of actions than the ones in adaptive query processing.

\noindent
{\bf Reinforcement Learning for Query Optimization.} Some query optimization methods aim at learning efficient query plans, e.g., order of joins, using reinforcement learning techniques \cite{trummer2019skinnerdb,ortiz2018learning,kaftan2018cuttlefish}. 
As opposed to our method, they do {\it not} learn query processing algorithms. 
They also face significantly fewer actions than ours.

\noindent{\bf Offline Learning for Query Processing.}
The authors in \cite{10.14778/3368289.336830} learn the distribution of the data to increase the sampling efficiency 
for queries with {\it count} aggregation function such that membership to the underlying relation can be checked efficiently. 
This approach supports a restricted class of queries. 
For example, it does {\it not} handle joins as checking the join condition for a tuple in one relation may
need a full scan of the other one, which requires a significant amount of time.
It also heavily relies on learning an accurate predictive model for some attributes based on the content of others. 
Such dependencies may not be always present. 
It first learns a model for the query and then uses it to sample tuples.
This separation of exploration and exploitation is shown to learn significantly less effective models in substantially longer time than the online methods that combine exploration and exploitation \cite{slivkins2019introduction}.
Our approach handles a significantly larger class of queries and uses collaborative online learning to find efficient strategies. 

\noindent{\bf Offline Learning to Optimize In-Memory Join.}
Researchers have also investigated using offline learning methods to optimize index-based, sort-merge and hash-based algorithms for in-memory join \cite{10.14778/3587136.3587148}.
As explained in Section~\ref{sec:intro}, these methods handle limited types of joins.
As opposed to our method, these optimization techniques are limited to in-memory joins.
As they use offline learning, they may not return subsets of results quickly.

\noindent{\bf Learned Probabilistic Structures.}
Instead of conventional precomputed data structures, such as indexes, researchers have leveraged a precomputed probabilistic graphical model over each relation to compute samples of their join \cite{10.1145/3448016.3457302}.
It takes a time linear to the size of a relation to learn such model over the relation.
This approach will face the challenges explained in Section~\ref{sec:intro} for methods that require precomputed data structures, e.g., limited types of join. 

\section{Conclusion and Future work}
\label{qp:conclusion}
Many users would like to investigate (subsets of) join results progressively over large and evolving data quickly.
Current methods rely on costly preparation and maintaining auxiliary data structures, time-consuming reorganization of input relations, or apply only on a limited type of joins. 
We proposed a novel approach to return (subsets of) joins progressively and efficiently over large dataset without any preprocessing.
In our approach, the scan operators in a join learn the portions of the tuples in each relation that produce the most join results during query execution quickly.
Our empirical study indicates that our proposed methods outperform similar current methods in many settings significantly and deliver similar performance in the rest.
We believe that our approach introduces a fresh path in query processing by treating each query operator as an online learning agent and modeling query processing as an online learning and collaboration problem. It allows us to leverage the existing research on online and multi-agent learning to design query processing methods. 
An interesting future direction is to learn query execution and optimization strategies online together.

\bibliographystyle{ACM-Reference-Format}
\bibliography{main}
\end{document}